\pdfoutput=1

\documentclass[a4paper,UKenglish,cleveref, autoref, thm-restate]{lipics-v2019}
\nolinenumbers

\usepackage{amsmath,amssymb,amsfonts,latexsym}
\usepackage{graphics}
\usepackage{url}
\usepackage{ifpdf}
\usepackage[noend,noline,ruled,scleft,nofillcomment]{algorithm2e}
\usepackage{color}
\usepackage{caption} 
\usepackage[sort]{cite}

\newcommand{\uc}{\texttt{Max Unique Coverage}\xspace
}

\newcommand{\mc}{\texttt{Max Coverage}\xspace
}

\newcommand{\setcover}{
  \texttt{Set Cover}\xspace
}

\newcommand{\etal}{{et al.~}}
\newcommand{\tO}{\tilde{O}}
\newcommand{\eps}{\epsilon}
\newcommand{\prob}[1]{ \Pr \left [ #1 \right ]}
\newcommand{\expec}[1]{ \textup{E} \left [ #1 \right ]}
\newcommand{\roundup}[1]{ \lceil  #1  \rceil}

\newcommand{\size}[1]{ \left|  #1   \right|}
\renewcommand{\exp}[1]{ \textup{exp} \left( #1   \right)}
\renewcommand{\O}{\mathcal{O}}

\DeclareMathOperator{\opt}{OPT}
\DeclareMathOperator{\polylog}{polylog}
\DeclareMathOperator{\poly}{poly}
\DeclareMathOperator{\Disj}{\texttt{DISJ}}

\title{Maximum Coverage in the Data Stream Model: 
Parameterized and Generalized} 

\titlerunning{} 

\author{Andrew McGregor}{University of Massachusetts, Amherst }{mcgregor@cs.umass.edu}{}{}

\author{David Tench}{Stony Brook University}{dtench@protonmail.ch}{}{}
\author{Hoa T. Vu}{San Diego State University}{hvu2@sdsu.edu}{}{}

\authorrunning{A.~McGregor, D.~Tench, and H.~T.~Vu} 

\Copyright{Andrew McGregor, David Tench, Hoa T. Vu} 

\ccsdesc[500]{Theory of computation~Sketching and sampling}
\ccsdesc[500]{Theory of computation~Approximation algorithms analysis}
\ccsdesc[500]{Theory of computation~Parameterized complexity and exact algorithms}

\keywords{Data streams, maximum coverage, maximum unique coverage, set cover} 

\relatedversion{} 

\funding{This work was partially supported by NSF grants CCF-1934846, CCF-1908849,  and   CCF-1637536.}

\hideLIPIcs  

\EventEditors{Ke Yi and Zhewei Wei}
\EventNoEds{2}
\EventLongTitle{24th International Conference on Database Theory (ICDT 2021)}
\EventShortTitle{ICDT 2021}
\EventAcronym{ICDT}
\EventYear{2021}
\EventDate{March 23--26, 2021}
\EventLocation{Nicosia, Cyprus}
\EventLogo{}
\SeriesVolume{186}
\ArticleNo{10}

\begin{document}

\maketitle


\begin{abstract}
We present algorithms for the $\mc$  and  $\uc$ problems in the data stream model.  The input to both problems are $m$ subsets of a universe of size $n$ and a value $k\in [m]$. In $\mc$, the problem is to find a collection of at most $k$ sets such that the number of elements covered by at least one set is maximized. In $\uc$, the problem is to find a collection of at most $k$ sets such that the number of elements covered by exactly one set is maximized. These problems are closely related to a range of graph problems including matching, partial vertex cover, and capacitated maximum cut.
 In the data stream model, we assume $k$ is given and the sets are revealed online. Our goal is to design single-pass algorithms that use space that is sublinear in the input size.  Our main algorithmic results are: 
\begin{itemize}
\item If the sets have size at most $d$, there exist single-pass algorithms using $O(d^{d+1} k^d)$ space that solve both problems exactly. This is optimal up to polylogarithmic factors for constant $d$. 
\item If each element appears in at most $r$ sets, we present single pass algorithms using $\tilde{O}(k^2 r/\epsilon^3)$ space that return a $1+\epsilon$ approximation in the case of  $\mc$. We also present a single-pass algorithm using slightly more memory, i.e., $\tilde{O}(k^3 r/\epsilon^{4})$ space, that $1+\eps$ approximates $\uc$.
\end{itemize}
In contrast to the above results, when $d$ and $r$ are arbitrary, any constant pass $1+\epsilon$ approximation algorithm for either problem requires $\Omega(\epsilon^{-2}m)$ space but  a single pass $O(\epsilon^{-2}mk)$ space algorithm exists. 
In fact any constant-pass algorithm with an approximation better than $e/(e-1)$ and $e^{1-1/k}$ for \mc and \uc respectively requires $\Omega(m/k^2)$ space when $d$ and $r$ are unrestricted.
  En route, we also obtain an algorithm for a parameterized version of the streaming $\setcover$ problem.
\end{abstract}

\maketitle

\section{Introduction}
\subparagraph{Problem Description.}  We consider the $\mc$ and $\uc$ problems in the data stream model. The input to both problems are $m$ subsets of a universe of size $n$ and a value $k\in [m]$. In $\mc$, the problem is to find a collection of at most $k$ sets such that the number of elements covered by at least one set is maximized. In $\uc$, the problem is to find a collection of at most $k$ sets such that the number of elements covered by exactly one set is maximized. In the data stream model, we assume $k$ is  provided but that the sets are revealed online and our goal is to design single-pass algorithms that use space that is sub-linear in the input size.

$\mc$ is a classic \texttt{NP-Hard} problem that has a wide range of applications including facility and sensor allocation \cite{KrauseG07}, information retrieval \cite{Anagnostopoulos15}, influence maximization in marketing strategy design \cite{KempeKT15}, and the blog monitoring problem \cite{SahaG09}. It is well-known that the greedy algorithm, which greedily picks the set that covers the most number of uncovered elements, is a $e/(e-1)$ approximation and that unless $\texttt{P}=\texttt{NP}$, this approximation factor is the best possible in polynomial time \cite{Feige98}.   

$\uc$ was first studied in the offline setting by Demaine \etal \cite{DemaineFHS08}. A motivating  application for this problem was in the design of wireless networks where we want to place base stations that cover mobile clients. Each station could cover multiple clients but unless a client is covered by a unique station the client would experience too much interference. Demaine \etal \cite{DemaineFHS08}  gave a polynomial time $O(\log k)$ approximation. Furthermore, they showed that $\uc$ is hard to approximate within a factor $O(\log^{\sigma} n)$ for some constant $\sigma$ under reasonable complexity assumptions. Erlebach and van Leeuwen \cite{ErlebachL08} and Ito \etal \cite{Ito14} considered a geometric variant of the problem and Misra et al.~\cite{MisraMRSS13} considered the parameterized complexity of the problem. This problem is also closely related to Minimum
Membership Set Cover where one has to cover every element and minimizes the maximum overlap on any element \cite{KuhnRWWZ05,DomGNW06}.

In the streaming set model, $\mc$ and the related $\setcover$ problem\footnote{That is, find the minimum number of sets that cover the entire universe.} have both received a significant amount of attention \cite{IndykMRUVY17,SahaG09,Har-PeledIMV16,ChakrabartiW16,EmekR16,AssadiKL16,MV18,IndykV19}.  The most relevant result is a single-pass $2+\epsilon$ approximation using $\tO(k \epsilon^{-3} )$ space \cite{MV18,BadanidiyuruMKK14} although better approximation is possible in a similar amount of space if multiple passes are permitted \cite{MV18} or if the stream is randomly ordered \cite{NTMZ18,ASS2020}. In this paper, we almost exclusively consider single-pass algorithms where the sets arrive in an arbitrary order.

The unique coverage problem has not  been studied in the data stream model although it, and $\mc$, are closely related to various graph problems that have been studied.

\subparagraph{Relationship to Graph Streaming.} There are two main variants of the graph stream model. In the \emph{arbitrary order model}, the stream consists of the edges of the graph in arbitrary order. In the \emph{adjacency list model}, all edges that include the same node are grouped together.
Both models generalize naturally to hypergraphs where each edge could consists of more than two nodes. The arbitary order model has been more heavily studied than the adjacency list model but there has still been a significant amount of work in the latter model \cite{McGregorVV16, McGregorV16, BravermanOV13, KonradMM12, Har-PeledIMV16,Assadi17,AssadiKL16,KallaugherMPV19,McGregorV20}.  For further details, see a recent survey on work on the graph stream model \cite{McGregor14}. 

To explore the relationship between $\mc$ and $\uc$  and various graph stream problems, it makes sense to introduce to additional parameters beyond $m$ (the number of sets) and $n$ (the size of the universe). Specifically, throughout the paper we let $d$ denote the maximum cardinality of a set in the input and let $r$ denote the maximum multiplicity of an element in the universe where the \emph{multiplicity} is the number of sets an element  appears.\footnote{Note that $d$ and $r$ are dual parameters in the sense that if the input is $\{S_1, \ldots, S_m\}$ and we define $T_i=\{j:i\in S_j\}$ then $d=\max_j |S_j|$ and $r=\max_i |T_i|$.} Then an input to $\mc$ and $\uc$ can define a (hyper)graph in one of the following two natural ways:
\begin{enumerate}
\item {\em First Interpretation:} A sequence of (hyper-)edges on a graph with $n$ nodes of maximum degree $r$ (where the degree of a node $v$ corresponds to how many hyperedges include that node) and $m$ hyperedges where each hyperedge has size at most $d$. In the case where every set has size $d=2$, the hypergraph is an \emph{ordinary graph}, i.e.,  a graph where every edge just has two endpoints. With this interpretation, the graph is being presented in the arbitrary order model.
\item {\em Second Interpretation:} A sequence of adjacency lists (where the adjacency list for a given node includes all the hyperedges that include that node) on a graph with $m$ nodes of maximum degree $d$ and $n$ hyperedges of maximum size $r$. In this interpretation, if every element appears in exactly  $r=2$ sets, then this corresponds to an ordinary graph where each element corresponds to an edge and each set corresponds to a node. With this interpretation, the graph is being presented in the adjacency list model.
\end{enumerate}

Under the first interpretation, the $\mc$ problem and the $\uc$ problem when all sets have size exactly $2$ naturally generalize the problem of finding a maximum matching in an ordinary graph in the sense that if there exists a matching with at least $k$ edges, the optimum solution to either $\mc$ and $\uc$ will be a matching. There is a large body of work on graph matchings in the data stream model \cite{AhnG13,EpsteinLMS09,mcgregor2005b,zelke,CrouchS14,CrouchMS13,mcgregor2005,KapralovKS14,Kapralov13,GoelKK12,KonradMM12,KonradR13,BuryS15a,Konrad15,GuruswamiO13} including work specifically on solving the problem exactly if the matching size is bounded \cite{ChitnisCHM15,ChitnisCEHMMV16}. More precisely, $\mc$ corresponds to the partial vertex cover problem \cite{Manurangsi19}: what is the maximum number of edges that can be covered by selecting $k$ nodes. For larger sets, the $\mc$ and $\uc$ are at least as hard as finding partial vertex covers and matching in hypergraphs.
 
Under the second interpretation, when all elements have multiplicity 2, then the problem $\uc$ corresponds to finding the capacitated maximum cut, i.e., a set of at most $k$ vertices such that the number of edges with exactly one endpoint in this set is maximized. In the offline setting, Ageev and Sviridenko \cite{AgeevS04} and Gaur \etal \cite{GaurKK11} presented a 2 approximation for this problem  using  linear programming and local search respectively. The (uncapacitated) maximum cut problem was been studied in the data stream model by Kapralov et al.~\cite{KapralovKS15,KapralovKSV17,KapralovK2018}; a 2-approximation is trivial in logarithmic space\footnote{It suffices to count the number of edges $M$ since there is always a cut whose size is at least $M/2$.} but improving on this requires space that is polynomial in the size of the graph.  The capacitated problem is a special case of the problem of  maximizing a non-monotone sub-modular function subject to a cardinality constraint. This general problem has been considered in the data stream model \cite{BadanidiyuruMKK14, ChekuriGQ15,ChakrabartiK15,HuangKY17} but in that line of work it is assumed that there is oracle access to the function being optimized, e.g., given any set of nodes, the oracle will return the number of edges cut. Alaluf et al.~\cite{AEFNS20}  presented a $2+\eps$ approximation in this setting, assuming exponential post-processing time. In contrast, our algorithm does not assume an oracle while obtaining a $1+\epsilon$ approximation (and also works for the more general problem $\uc$).

\subsection{Our Results} 

Our main results are the following single-pass  streaming algorithms\footnote{Throughout we use $\tilde{O}$ to denote that logarithmic factors of $m$ and $n$ are being omitted.}:
\begin{description}
\item[(A) Bounded Set Cardinality.] If all sets have size at most $d$, there exists a $\tO(d^{d+1}k^d)$ space data stream algorithm that solves $\uc$ and $\mc$ exactly. We show that this is nearly optimal in the sense that any exact algorithm requires $\Omega(k^d)$ space for constant $d$. 
\item[(B) {Bounded Multiplicity.}] If every element appears in at most $r$ sets, we present the following algorithms:
\begin{itemize}
 \item  (B1) $\uc$: There exists a $1+\epsilon$ approximation using $\tilde{O}(\epsilon^{-4}k^3 r)$ space.  
\item  (B2) $\mc$: There exists a $1+\epsilon$ approximation algorithm using  $\tilde{O}(\epsilon^{-3}k^2 r)$ space.
\end{itemize}

\end{description}
In contrast to the above results, when $d$ and $r$ are arbitrary, any constant pass $1+\epsilon$ approximation algorithm for either problem requires $\Omega(\epsilon^{-2}m)$ space \cite{Assadi17}.\footnote{The lower bound result by Assadi \cite{Assadi17} was for the case of $\mc$ but we will explain that it also applies in the case of $\uc$.} We also generalize of lower bound for $\mc$ \cite{MV18} to $\uc$ to show that any constant-pass algorithm with an approximation better than $e^{1-1/k}$ requires $\Omega(m/k^2)$ space. We also present a single-pass algorithm with an $O(\log \min(k,r))$ approximation for $\uc$ using $\tO(k^2)$ space, i.e., the space is independent of $r$ and $d$ but the approximation factor depends on $r$. This algorithm is a simple combination of a $\mc$ algorithm due to McGregor and Vu~\cite{MV18} and an algorithm for $\uc$ in the offline setting due to Demaine et al.~\cite{DemaineFHS08}.
Finally, our $\mc$ result (B2) algorithm also yields a new multi-pass result for a parameterized version of the streaming $\setcover$ problem. We will also show that results (A) and (B2) can also be made to handle stream deletions. The generalization for result (A) that we present requires space that scales with $k^{2d}$ rather than $k^d$. However, in subsequent work we have shown that space the scales with $k^d$ is also sufficient in the insert/delete setting.

\subsection{Technical Summary and Comparisons} 
\subparagraph{Technical Summary.} Our results are essentially streamable kernelization results, i.e., the algorithm ``prunes'' the input (in the case of $\uc$ and $\mc$ this corresponds to ignoring some of the input sets) to produce a ``kernel''  in such a way that a) solving the problem optimally on the kernel yields a solution that is as good (or almost as good) as the optimal solution on the original input and b)  the kernel can be constructed in the data stream model 
 and is sufficiently smaller than the original input such that it is possible to find an optimal solution for the kernel in significantly less time than it would take to solve on the original input. In the field of fixed parameter tractability, the main requirement is that the kernel can be produced in polynomial time. In the growing body of work on streaming kernelization \cite{ChitnisCEHMMV16,ChitnisCEHM15,ChitnisC19} the main requirement is that the kernel can be constructed using small space in the data stream model. Our results fits in with this line of work and the analysis requires numerous combinatorial insights into the structure of the optimum solution for $\uc$ and $\mc$. 

Our technical contributions can be outlined as follows.
\begin{itemize}
\item Result (A) relies on a key combinatorial lemma. This lemma provides a rule to discard sets such that there is an optimum solution that does not contain any of the discarded sets. Furthermore, the number of stored sets can be bounded in terms of $k$ and $d$.
\item Result (B1) uses the observation that each set of any optimal solution intersects some maximal collection of disjoint sets.  The main technical step is to demonstrate that storing a small number of intersecting sets, in terms of $k$ and $r$, suffices to preserve the optimal solution. 
\item Result (B2) is based on a very simple idea of first collecting the largest $O(rk/\epsilon)$ sets and then solving the problem optimally on these sets. This can be done in a space efficient manner using existing sketch for $F_0$ estimation in the case of $\mc$. While the approach is simple, showing that it yields the required approximations requires some work that builds on a recent result by Manurangsi \cite{Manurangsi19}. We also extend the algorithm to the model where sets can be inserted and deleted. 
\end{itemize}

\subparagraph*{Comparison to Related Work.} In the context of streaming algorithms, for the $\mc$ problem, McGregor and Vu \cite{MV17} showed that any approximation better than $e/(e-1)$ requires $\Omega(m/k^2)$ space. For the more general problem of streaming submodular maximization subject to a cardinality constraint, Feldman et al. \cite{FeldmanNSZ20} very recently showed a stronger lower bound that any approximation better than 2 requires $\Omega(m)$ space. Our results  provide a route to circumvent these bounds via parameterization on $k,r,$ and $d$. 

Result (B2) also leads to a parameterized algorithm for streaming $\setcover$. This new algorithm uses $\tO(rk^2 n^{\delta} + n)$ space which improves upon the algorithm by Har-Peled et al.~\cite{Har-PeledIMV16} that uses $\tO(mn^{\delta} + n)$ space, where $k$ is an upper bound for the size of the minimum set cover,  in the case $rk^2 \ll m$. Both algorithms use $O(1/\delta)$ passes and yield an $O(1/\delta)$ approximation.

In the context of offline parameterized algorithms, Bonnet et al. \cite{BonnetPS16} showed that $\mc$ is fixed-parameter tractable in terms of $k$ and $d$. However, their branching-search algorithm cannot be implemented in the streaming setting. Misra et al. \cite{MisraMRSS13} showed that the maximum unique coverage problem in which the aim is to maximize the number of uniquely covered elements $u$ (without any restriction on the number of sets) admits a kernel of size $4^u$. On the other hand, they showed that the budgeted version of this problem (where each element has a profit and each set has a cost and the goal is maximize the profit subject to a budget constraint) is $W[1]$-hard when parameterized by the budget\footnote{In the $\uc$ problem that we consider, all costs and profits are one and the budget is $k$.}. In this context, our result shows that a parameterization on both the maximum set size $d$ and the budget $k$ is possible (at least when all costs and profits are unit).


\section{Preliminaries}
\label{sec:prelim}
\subsection{Notation and Parameters}
Throughout the paper, $m$ will denote the number of sets,  $n$ will denote the size of the universe, and $k$ will denote the maximum number of sets that can be used in the solution.
Given input sets $S_1, S_2, \ldots, S_m\subseteq [n]$, let \[d=\max_i |S_i|\] be the maximum set size and let \[r=\max_j |\{i:j\in S_i\}|\] be the maximum number of sets that contain the same element.

Suppose $C$ is a collection of sets. We let $F(C)$ (and $G(C)$) be the set of elements covered (and uniquely covered) by an optimal solution in $C$. Furthermore, let $f(C) = |F(C)|$ and $g(C)=|G(C)|$. In other words, $f(C)$ is the maximum number of elements that can be covered by $k$ sets. Similarly,  $g(C)$ is the maximum number of elements that can be uniquely covered by $k$ sets. Furthermore, let $\psi(C)$ and $\tilde{\psi}(C)$ be the set of elements covered and uniquely covered respectively by the sets in $C$.

To ease the notation, if $C$ is a collection of set and $S$ is a set, we often use $C-S$ to denote $C \setminus \{S\}$ and $C+S$ to denote $C \cup \{S\}$.

We use $M$ to denote the collection of all sets in the stream. Therefore, the optimal value to $\mc$ and $\uc$ are $f(M)$ and $g(M)$ respectively.

Throughout this paper, we say an algorithm is correct with high probability if the probability of failure is inversely polynomial in $m$.
 
 \subsection{Sketches and Subsampling}\label{sec:sketch}

\subparagraph{Coverage Sketch.} Given a vector $x\in {\mathbb R}^n$, $F_0(x)$ is defined as the number of elements of $x$ which are non-zero. If given a subset $S\subset \{1, \ldots, n\}$, we define  $x_S\in \{0,1\}^n$ to be the characteristic vector of $S$ (i.e., $x_i=1$ iff $i\in S$) then given sets $S_1, S_2, \ldots$ note that $F_0(x_{S_1} + x_{S_2}+\ldots )$ is exactly the number of elements covered by $S_1\cup S_2\cup \ldots $. We will use the following  result for estimating $F_0$.

 \begin{theorem}[$F_0$ Sketch \cite{CormodeDIM03, BJKST02}]\label{thm:F0-approximation}
Given a set $S\subseteq [n]$, there exists an $\tO(\epsilon^{-2}\log \delta^{-1})$-space algorithm that  constructs a data structure $\mathcal{M}(S)$ (called an \emph{$F_0$ sketch} of $S$). The sketch has the property that the number of distinct elements in a collection of sets $S_1, S_2, \ldots, S_t$ can be approximated up to a $1 + \epsilon$ factor with  probability at least $1-\delta$ provided the collection of $F_0$ sketches $\mathcal{M}(S_1), \mathcal{M}(S_2), \ldots, \mathcal{M}(S_t)$.
\end{theorem}

Note that if we set $\delta\ll 1/(\poly(m) \cdot \binom{t}{k})$ in the above result we can try each collection of $k$ sets amongst  $S_1, S_2, \ldots, S_t$ and get a $1+\epsilon$ approximation for the coverage of each collection with high probability.

\label{sec:subsample} 

\subparagraph{Unique Coverage Sketch.} For unique coverage, our sketch of a set corresponds to subsampling the universe via some hash function $h:[n]\rightarrow \{0,1\}$ where $h$ is chosen randomly such that for each $i$, $\prob{h(i)=1}=p$ for some appropriate value $p$. Specifically, rather processing an input set $S$, we process $S'=\{i\in S: h(i)=1\}$. Note that $|S'|$ has size $p|S|$ in expectation. This approach was use by McGregor and Vu \cite{MV18} in the context of $\mc$ and it extends easily to $\uc$; see Section \ref{appendix:subsampling}. The consequence is that if there is a streaming algorithm that finds a $t$ approximation, we can turn that algorithm into a $t(1+\epsilon)$ approximation algorithm in which we can assume that $\opt = O(\epsilon^{-2} k \log m)$ with high probability by running the algorithm on a subsampled sets rather than the original sets. Note that this also allows us to assume input sets have size $O(\epsilon^{-2} k \log m)$ since $|S'|\leq \opt$. Hence each ``sketches'' set can be stored using $B=O(\epsilon^{-2} k \log m\log n)$ bits.

\subparagraph{An Algorithm with $\tilde{O}(\epsilon^{-2} mk)$ Memory.}
We will use the above sketches in a more interesting context later in the paper,  but note that they immediately imply a trivial algorithmic result. Consider the naive algorithm that stores every set and finds the best solution; note that this requires exponential time.  We note that since we can assume $\opt = O(\epsilon^{-2} k \log m)$, each set has size at most $O(\epsilon^{-2} k \log m)$. Hence, we need $\tO(\epsilon^{-2} m k)$ memory to store all the sets. This approach was noted in \cite{MV18} in the context of $\mc$ but also apples to $\uc$.
 We will later show that for a $1+\epsilon$ approximation, the above trivial algorithm is optimal up to polylogarithmic factors for constant $k$.


\section{An Exact Algorithm}\label{subsec:exact_alg}

\subparagraph{Algorithm.} Our algorithm, though perhaps non-intuitive, is simple to state:

\begin{enumerate}
\item Initialize $X$ to be an empty collection of sets. Let $b=d(k-1)$.
\item Let $X_a$ be the  sub-collection of $X$ that contains sets of size $a$. 
\item For each set $S$ in the stream: Suppose $|S|=a$. Add $S$ to $X$ if there does not exist  $T \subseteq S$ that occurs as a subset of $(b+1)^{d-|T|}$ sets of $X_a$. 
\item Post-processing: Return the best solution $C$ in $X$.
\end{enumerate}

\subparagraph{Analysis.} Our algorithm relies on the following  combinatorial lemma.

\begin{lemma}\label{lem:supsub}
Let $W=\{S_1, S_2, \ldots \}$ be a collection of distinct sets where each $S_i\subseteq [n]$ and $|S_i|=a$. Suppose for all $T \subseteq \psi(W)$ with $|T|\leq a$ there exist at most 
\[\ell_{|T|}:=(b+1)^{a-|T|}\] sets  in $W$ that contain $T$. Furthermore, suppose  there exists a set $T^*$ such that this inequality is tight. Then, for all $B\subseteq \psi(W)$ disjoint from $T^*$ with $|B|\leq b$ there exists a set $Y \in W$ such that $T^*\subseteq Y$ and $|Y\cap B|=0.$
\end{lemma}
\begin{proof}
If $|T^*|=a$ then $T^*\in W$, then we can simply set $Y = T^*$. Henceforth, assume $|T^*|<a$.
Consider the $\ell_{|T^*|}$ sets in $W$ that are supersets of $T^*$. Call this collection $W'$.
For any $x \in B$, there are at most $\ell_{|T^*|+1}$ sets that include $T^* \cup \{x\}$. Since there are $b$ choices for $x$, at most 
\[
b\ell_{|T^*|+1}=b (b+1)^{a-|T^*|-1}<(b+1)^{a-|T^*|}= \ell_{|T^*|} 
\]
sets in $W'$ contain an element in $B$. Hence, at least one set $Y$ in $W'$ does not contain any element in $B$. 
\end{proof}

We show that the algorithm indeed obtains an exact kernel for the problems. Recall that $M$ is the collection of all sets in the stream, i.e., the optimal solution has size $f(M)$.

\begin{theorem}
\label{thm:exact}
The output of the algorithm is optimal. In particular, $f(C) = f(M)$ and $g(C) = g(M)$.
\end{theorem}
\begin{proof}
Recall that $X$ is the collection of all stored sets. We define
\begin{align*}
C_i & = M \setminus \{\text{the first $i$ sets in the stream that are not stored in $X$}\}.
\end{align*}

Clearly, $f(C_0)=f(M)$. Now, suppose there exists $i \geq 1$ such that $f(C_{i}) < f(M)$. Let $i$ be the smallest such index. Let $\O$ be an optimal solution of $C_{i-1}$ (note that $\O$ is also an overall optimal solution based on the minimal assumption on $i$). Let $S$ be the $i$th set that was not stored in $X$. If $S \notin \O$ then we have a contradiction since $f(C_i) = f(C_{i-1})=f(M)$. Thus, assume $S \in \O$. Suppose $|S|=a$.

\begin{claim}\label{claim:exact}
There exists $Y$ in $X_a$ such that $f(\O-S+Y) \geq f(\O)$.
\end{claim}
\begin{proof}
Note that $S$ was not stored because there existed  $T^* \subseteq S $ such that $T^*$ was a subset of $(b+1)^{d-|T^*|}$ sets in $X_a$. Consider the set $B= \psi(\O) \setminus S$. Clearly, $B \cap T^* =\emptyset$ and $|B| \leq d(k-1)$. 
By Lemma \ref{lem:supsub}, there is a set $Y$ in $X_a$ such that  $Y \cap B = \emptyset$. 

Let $Y'=Y\setminus S$ and $S'=S\setminus Y.$ Note that $|Y'|=|S'|$ since $|Y|=|S|$. Define indicator variables $\alpha_z=1$ iff $z\in \psi(\O -S+Y)$ and  $\beta_z=1$  iff $z\in \psi(\O)$. Note that 
\begin{align*}
\left (z\in Y\cap S  \mbox{ or } z\not \in Y\cup S\right ) \implies \left (\alpha_z=\beta_z \right ), \\
\left (z\in Y' \right )\implies \left (\alpha_z=1 \right ), \\
\left (z\in Y' \right ) \implies \left ( \beta_z=0 \right ),
\end{align*}
where the last equation uses the fact that $Y'$ is disjoint from $\psi(\O)$.
Then 
\begin{align*}
 |\psi( \O -S+Y)|  = &   \sum_{z\in Y'}\alpha_z +\sum_{z\in Y\cap S}\alpha_z+\sum_{z\in S'}\alpha_z+\sum_{z\not \in Y\cup S}\alpha_z\\
\geq &   \left (|Y'|+\sum_{z\in Y'} \beta_z\right ) +\sum_{z\in Y\cap S}\beta_z+\left (-|S'|+\sum_{z\in S'}\beta_z \right )+\sum_{z\not \in Y\cup S}\beta_z\\
= &   \sum_{z\in Y'}\beta_z +\sum_{z\in Y\cap S}\beta_z+\sum_{z\in S'}\beta_z+\sum_{z\not \in Y\cup S}\beta_z= |\psi( \O)| ~.\qedhere
\end{align*}
\end{proof}

Thus, $f(C_i) \geq f(\O)=f(M)$ which is a contradiction. Hence, there is no such $i$ and the claim follows. The proof for unique coverage is almost identical: for the analogous claim we define indicator variables $\tilde{\alpha}_z=1$ iff $z\in \tilde{\psi}(\O -S+Y)$ and  $\tilde{\beta}_z=1$ iff $z\in \tilde{\psi}(\O)$. The proof goes through with $\alpha$ and $\beta$ replaced by $\tilde{\alpha}$ and $\tilde{\beta}$ since it is still the case that  
\begin{align*}
\left (z\in Y\cap S  \mbox{ or } z\not \in Y\cup S\right ) \implies \left (\tilde{\alpha}_z=\tilde{\beta}_z \right ),\\
\left (z\in Y' \right )\implies \left (\tilde{\alpha}_z=1 \right ),\\
\left (z\in Y' \right ) \implies \left ( \tilde{\beta}_z=0 \right ),
\end{align*}
where now the last two equations use the fact that $Y'$ is disjoint from $\psi(\O)$. 
%
\end{proof}

\begin{lemma}
The space used by the algorithm is $\tilde{O}(d^{d+1}k^{d})$.
\end{lemma}
\begin{proof} Recall that one of the requirements for a set $S$ to be added to $X$ is that the number of sets in $X_{|S|}$ that are supersets of any subset of $S$ of size $t$ is at most $(b+1)^{d-t}$. This includes the empty subset and since every set in $X_{|S|}$ is a superset of the empty set, we deduce that $|X_{|S|}|\leq  (b+1)^{d} = O((dk)^d)$. Since each set needs $\tilde{O}(d)$ bits to store, and $|X|=\sum_{a=1}^d |X_a| \leq O(d^{d}k^d)$, the total space is $\tilde{O}(d^{d+1} k^d)$.
\end{proof}

We summarize the above as a theorem.

\begin{theorem}\label{theorem:exact}
There exist deterministic single-pass algorithms using $\tO(k^d d^{d+1} )$ space that yields an exact solution to $\mc$ and $\uc$.
\end{theorem}

\subparagraph{Handling Insertion-Deletion Streams.}  We outline another exact algorithm that works for insertion-deletion streams, however with a worse space bound $\tO((kd)^{2d})$, in Section \ref{sec:deletion-streams-exact}.
\begin{theorem}\label{theorem:deletion-exact}
There exist randomized single-pass algorithms using $\tO(d^{2d}k^d)$ space and allowing deletions that w.h.p.~yield an exact solution to $\mc$ and $\uc$.
\end{theorem} \footnote{This improves upon our earlier result in the ICDT version of the paper that uses $\tO((dk)^{2d})$ space.}
\section{Approximation Algorithms}

In this section, we present a variety of different approximation algorithms where the space used by the algorithm is independent of $d$ but, in some cases, may depend on $r$. The first algorithm uses $\tilde{O}(\epsilon^{-4}k^3 r)$ memory and obtains a $1+\epsilon$ approximation to both problems. The second algorithm uses $\tilde{O}(\epsilon^{-3}k^2 r)$ memory and obtains a $1+\epsilon$ approximation to \mc and a $2+\epsilon$ approximation to \uc; it can also be extended to streams with deletions.

\subsection{A $1+\epsilon$ Approximation}

Given a collection of sets $C=\{S_1, S_2, \ldots, S_m\}$, we say a sub-collection $C'\subset C$ is a \emph{matching} if the sets in $C'$ are mutually disjoint. $C'$ is a maximal matching if there does not exist $S\in C\setminus C'$ such that $S$ is disjoint from all sets in $C'$. 
\begin{lemma}\label{lem:matching}
For any input $C$, let $O\subset C$ be an optimal solution for either the $\mc$ or $\uc$ problem. Let $M_i$ be a maximal matching amongst the input set of size $i$. Then every set of size $i$ in $O$ intersects with some set in $M_i$.
\end{lemma}
\begin{proof}
Let $S\in O$ have size $i$. If it was disjoint from all sets in $M_i$ then it could be added to $M_i$ and the resulting collection would still be a matching. This violates the assumption that  $M_i$ is maximal.
\end{proof}

The next lemma extends the above result  to show that we can potentially remove many sets from each $M_i$ and still argue that there is an optimal solution for the original instance amongst the sets that intersect a set in some $M_i$.

\begin{lemma}\label{lem:matching2}
Consider an input of sets of size at most $d$. For $i\in [d]$, let $M_i$ be a maximal matching amongst the input set of size $i$ and let $M_i'$ be an arbitrary subset of $M_i$ of size $\min(k+dk,|M_i|)$. Let $D_i$ be the collection of all sets that intersect a set in $M_i'$. Then $\bigcup_i (D_i\cup M_i')$ contains an optimal solution to both the $\uc$ and $\mc$ problem. 
\end{lemma}
\begin{proof}
If $|M_i|=|M'_i|$ for all $1\leq i\leq d$ then the result follows from Lemma \ref{lem:matching}.  If not, let $j=\max \{i\in [d]: |M_i|>|M'_i|\}$. Let $\O$ be an optimal solution and let $\O_i$ be all the sets in $\O$ of size $i$. We know that every set in $\O_d\cup \O_{d-1} \cup \ldots \cup \O_{j+1}$ is in 
\[\bigcup_{i\geq j+1} (D_i\cup M_i')=\bigcup_{i\geq j+1} (D_i\cup M_i) \ .\] 
Hence, the number of elements (uniquely) covered by $\O$ is at most the number of elements (uniquely) covered by $\O_d\cup \O_{d-1} \cup \ldots \cup \O_{j+1}$ plus $kj$ since every set in $\O_j\cup \ldots \cup \O_1$ (uniquely) covers at most $j$ additional elements. But we can (uniquely) cover at least the number of elements (uniquely) covered by $\O_d\cup \O_{d-1} \cup \ldots \cup \O_{j+1}$ plus $kj$. This is  because $M_j$ contains $k+dk$ disjoint sets of size $j$ and at least $k+dk-kd=k$  of these are disjoint from all sets in $\O_d\cup \O_{d-1} \cup \ldots \cup \O_{j+1}$. Hence, there is a solution amongst $\bigcup_{i\geq j} (D_i\cup M_i')$ that is at least as good as $\O$ and hence is also optimal.
\end{proof}

The above lemma suggests  an exact algorithm that stores the sets in $\bigcup_i (D_i\cup M_i')$ and find the optimum solution among these sets. In particular,  we construct matchings of each size greedily up to the appropriate size and store all intersecting sets. Note that since each element belongs to at most $r$ sets, the total space is $\tO(d^2 k r)$. Applying the sub-sampling framework, we have $d \leq \opt = O(k/\epsilon^2 \log m)$ and the approximation factor becomes $1+\epsilon$. 

\begin{theorem}
There exists a randomized one-pass algorithm using $\tilde{O}(\epsilon^{-4}k^3 r)$ space that finds a $1+\epsilon$ approximation to  $\uc$ and $\mc$.
\end{theorem}
 
\subsection{A More Efficient $1+\epsilon$ Approximation for Maximum Coverage } \label{sec:manu}

In this section, we generalize the approach of Manurangsi \cite{Manurangsi19} and combine that with the $F_0$-sketching technique  to obtain a $1+\epsilon$ approximation using $\tO(\epsilon^{-3} k^2r )$ space for maximum coverage. This saves a factor $k/\epsilon$ and the generalized analysis might be of independent interest. Let $\opt = \psi(\O)$ denote the optimal coverage of the input stream.

Manurangsi  \cite{Manurangsi19} showed that for the maximum $k$-vertex cover problem, the $\Theta(k/\epsilon)$ vertices with highest degrees form a $1+\epsilon$ approximation kernel for the maximum $k$ vertex coverage problem. That is, there exist $k$ vertices among those that cover $(1-\epsilon)\opt$ edges. We now consider a set system in which an element belongs to at most $r$ sets (this can also be viewed as a hypergraph where each set corresponds to a vertex and each element corresponds to a hyperedge; we then want to find $k$ vertices that touch as many hyperedges as possible).

We begin with the following lemma that generalizes the aforementioned result in \cite{Manurangsi19}. We may assume that $m > rk/\epsilon $ since otherwise, we can store all the sets. 
\begin{lemma}
Suppose $m > \roundup{rk/\epsilon}$. Let $K$ be the collection of $\roundup{rk/\epsilon}$ sets with largest sizes (tie-broken arbitrarily). There exist $k$ sets in $K$ that cover $(1-\epsilon)\opt$ elements.
\end{lemma}
\begin{proof}
Let $\O$ denote the collection of $k$ sets in some optimal solution. Let $\O^{in} = \O \cap K$ and $\O^{out} = \O \setminus K$. We consider a random subset $Z \subset K$ of size $|\O^{out}|$. We will show that  the sets in $Z \cup \O^{in}$ cover $(1-\epsilon)\opt$ elements in expectation; this implies the claim.  

Let $[\mathcal{E}]$ denote the indicator variable for event $\mathcal{E}$.  We rewrite
\begin{align*}
|\psi(Z \cup \O^{in})| = |\psi(\O^{in})| + |\psi(Z)| -  |\psi(\O^{in})\cap \psi(Z)|~.
\end{align*}

Furthermore, the probability that we pick a set $S$ in $K$ to add to $Z$ is
\[
p := \frac{|\O^{out}|}{|K|} \leq \frac{k}{kr/\epsilon  } = \frac{\epsilon}{r}~.
\]

Next, we upper bound $\expec{|\psi(\O^{in}) \cap \psi(Z)|}$. We have
\begin{align*}
\expec{|\psi(\O^{in}) \cap \psi(Z)|} & \leq \sum_{u \in \psi(\O^{in})} \sum_{S \in K: u \in S} \prob{S \in Z}  \leq \sum_{u \in \psi(\O^{in})}r p  \leq |\psi(\O^{in})| \cdot  \epsilon ~.
\end{align*}
We lower bound $ \expec{|\psi(Z)|}$ as follows.
\begin{align}  \expec{|\psi(Z)|} \geq~ &\expec{   \sum_{S \in K}  \left( |S| [S \in Z] - \sum_{S' \in K \setminus \{S\}}|S \cap S'| [S \in Z \land S' \in Z] \right)} \nonumber \\
 \geq~& \sum_{S \in K} \left( |S| p -   \sum_{S' \in K \setminus \{S\}}|S \cap S'| p^2  \right) \nonumber\\
   \geq~ &\sum_{S \in K} \left( |S| p -   (r-1)|S| p^2  \right) \geq  p(1-pr)\sum_{S \in K} |S|  \geq  p(1-\epsilon)\sum_{S \in K} |S| \label{eq:der}~.
\end{align}

In the above derivation, the second inequality follows from the observation that 
\[\prob{S \in Z \land S' \in Z} \leq p^2~.\] The third inequality is because  $\sum_{S' \in K \setminus \{S\}}|S \cap S'| \leq (r-1) |S|$ since each element belongs to at most $r$ sets.

For all $S \in K$, we  must have  
\[
|S| \geq \frac{\sum_{Y \in \O^{out}}|Y|}{|\O^{out}|}  \geq \frac{|\psi(\O^{out})|}{|\O^{out}|} ~.
\]
Thus, 
\begin{align*}
\expec{|\psi(Z)|}  
 \geq  p \left(1- \epsilon \right)  |K|\frac{|\psi(\O^{out})|}{|\O^{out}|} 
&=  p \left(1- \epsilon \right)  \frac{|\psi(\O^{out})|}{p} = (1-{\epsilon})  |\psi(\O^{out})|~.
\end{align*}
Putting it together,
\begin{align*}
\expec{|\psi(Z \cup \O^{in})|} & \geq  |\psi(\O^{in})| +  (1-\epsilon)  |\psi(\O^{out})|  - |\psi(\O^{in})|  \cdot  \epsilon  \geq (1-\epsilon)\opt~. \qedhere
\end{align*} 
\end{proof}

With the above lemma in mind, the following algorithm's correctness is immediate.
\begin{enumerate}
\item Store $F_0$-sketches of the $\roundup{kr/\epsilon}$ largest sets, where the failure probability of the sketches is set to $\frac{1}{ \poly(n) {m \choose k}}$. 
\item At the end of the stream, return the $k$ sets with the largest coverage based on the estimates given by the $F_0$-sketches.
\end{enumerate}

We restate our result as a theorem.

\begin{theorem}\label{thm:alg6}
There exists a randomized one-pass, $\tO(k^2 r/\epsilon^3 )$-space, algorithm that with high probability finds a $1+\epsilon$ approximation to $\mc$.
\end{theorem}

\subparagraph{Obtaining a $2+\epsilon$ approximation to \uc.} We note that finding the best solution to \uc in $K$ will yield a $2+\epsilon$ approximation. This is a worse approximation than that of the previous subsection. However, we save a factor of $k/\epsilon$ in memory. Furthermore, this approach also allows us to handle streams with deletions.

To see that we get a $2+\epsilon$ approximation to \uc. Note that $g(Z \cup \O^{in}) \geq \frac{1}{2} \left( g(\O^{in}) + g(Z) \right)$. Furthermore, a similar derivation shows $\expec{|\tilde{\psi}(Z)|} \geq (1-\epsilon) |\tilde{\psi} (\O^{out})|$. Specifically, in the derivation in Eq.~\ref{eq:der}, we can simply replace $\psi$ with $\tilde{\psi}$. This gives us $g(K) \geq (1/2-\epsilon) g(\O)$. 


\subparagraph{Extension to Insert/Delete Streams.} The result can be extended to the case where sets are inserted and deleted. For the full details, see Section \ref{sec:deletion-streams-2}.


\subsection{An $O(\log \min(k,r))$ Approximation for Unique Coverage}

We now present an algorithm whose space does not depend on $r$ but the result comes at the cost of increasing the approximation factor to $O(\log(\min(k,r)))$. It also has the feature that the running time is polynomial in $k$ in addition to being polynomial in $m$ and $n$.

The basic idea is as follows: We consider an existing algorithm that first finds a 2.01 approximation $C$ to  $\mc$. It then finds the best solution of $\uc$ among the sets in $C$. 

\begin{theorem}\label{thm:alg4}
There exists a randomized one-pass, $\tO(k^2)$-space, algorithm that with high probability finds a $O(\log \min(k,r))$ approximation to $\uc$.
\end{theorem}
\begin{proof}
From previous work \cite{MV18, BadanidiyuruMKK14}, we can find a $2.01$ approximation $C$ to $\mc$ using $\tO(k)$ memory. Note that their algorithm maintains a collection $C$ of $k$ sets during the stream. Demaine \etal \cite{DemaineFHS08} proved that that if $Q$ is the best solution to $\uc$ among the sets in $C$, then $Q$ is an $O(\log \min(k,r)) $ approximation to $\uc$. In fact, they presented a polynomial time algorithm to find $Q$ from $C$ such that the number of uniquely covered elements is at least 
\[
\Omega(1/\log k) \cdot \size{\psi(C)}  \geq \Omega(1/\log k) \cdot 1/2.01  \cdot f(M) \geq \Omega(1/\log k)  \cdot g(M) ~.
\]
Note that storing each set in $C$ requires $\tO(d)$ memory. Hence, the total memory is $\tO(kd)$. Applying the sub-sampling framework, we obtain an $\tO(k^2)$ memory algorithm.

\end{proof}

\subsection{Application to  Parameterized Set Cover} \label{sec:setcover}
We parameterize the set cover problem as follows. Given a set system, either A) output a set cover of size $\alpha k$ if $\opt \leq k$ where $\alpha$ the approximation factor or B) correctly declare that a set cover of size $k$ does not exist. 

\begin{theorem}
For $0 < \delta <1$, there exists a randomized,~ $O(1/\delta)$-pass, $\tO(rk^2 n^{\delta} + n)$-space, algorithm that with high probability finds a $O(1/\delta)$ approximation to the parameterized \setcover problem.
\end{theorem}
\begin{proof}
In each pass, we run the algorithm in Theorem \ref{thm:alg6} with parameters $k$ and $\epsilon = 1/n^{\delta/3}$ on the remaining uncovered elements.  The space use is $\tO(rk^2 n^{\delta} + n)$. Here, we need additional $\tO(n)$ space to keep track of the remaining uncovered elements. 

Note that if $\opt \leq k$, after each pass, the number of uncovered elements is reduced by a factor $1/n^{\delta/3}$. This is because if $n'$ is the number of uncovered elements at the beginning of a pass, then after that pass, we cover all but at most $ n'/n^{\delta/3}$ of those elements. After $i$ passes, the number of remaining uncovered elements is $O(n^{1-i\delta/3})$;  we therefore use at most $O(1/\delta)$ passes until we are done. At the end, we have a set cover of size $O(k/\delta)$. 

If after $\omega(1/\delta)$ passes, there are still remaining uncovered elements, we declare  that such a solution does not exist. 
\end{proof}

Our algorithm improves upon the algorithm by Har-Peled et al.~\cite{Har-PeledIMV16} that uses $\tO(mn^{\delta} + n)$ space for when $rk^2 \ll m$. Both algorithms yield an $O(1/\delta)$ approximation and use  $O(1/\delta)$ passes.

\section{Lower Bounds} 

\subsection{Lower Bounds for Exact Solutions}
As observed earlier, any exact algorithm for either the $\mc$ or $\uc$ problem on an input where all sets have size $d$ will return a matching of size $k$ if one exists. However, by a lower bound due to Chitnis et al.~\cite{ChitnisCEHMMV16} we know that determining if there exists a matching of size $k$ in a single pass requires $\Omega(k^d)$ space. This immediately implies the following theorem. 

\begin{theorem}
Any single-pass algorithm that solves $\mc$ or $\uc$ exactly with probability at least $9/10$ requires $\Omega(k^d)$ space.
\end{theorem}

\subsection{Lower bound for a $e^{1-1/k}$ approximation} The strategy is similar to previous work on $\mc$ \cite{MV17,MV18}.  However, we need to argue that the relevant probabilistic construction works for all collections of fewer than $k$ sets  since the unique coverage function is not monotone. 

We make a reduction from the communication problem  $k$-player set disjointness, denoted by  $\Disj(m,k)$. In this problem, there are $k$ players where the  $i$th player has a set $S_{i} \subseteq [m]$. It is promised that exactly one of the following two cases happens a) NO instance: All the sets are pairwise disjoint and b)  YES instance: There is a unique element $v \in [m]$ such that $v \in S_i$ for all $i \in [k]$ and all other elements belong to at most one set. The (randomized) communication complexity (in the one-way model or the blackboard model), for some large enough constant success probability, of the above problem is $\Omega(m/k)$  even if the players may use public randomness \cite{ChakrabartiKS03}.   We can assume that $|S_1 \cup S_2 \cup \ldots \cup S_k| \geq m/4$ via a padding argument.

 \begin{theorem}\label{thm:lower-bound-1}
 Any constant-pass randomized algorithm with an approximation better than $e^{1-1/k}$ to $\uc$ requires $\Omega(m/k^2)$ space.
 \end{theorem}
 \begin{proof}
For each $i \in [m]$, let $\mathcal{P}_i$ be a random partition of $[n]$ into $k$ sets $V^i_1,\ldots,V^i_k$ such that an element in the universe $U = [n]$ belongs to exactly one of these sets uniformly at random. In particular, for all $i \in [m]$ and $v \in U$, 
 \[
 \prob{v \in V^i_j \land (\forall j' \neq j, v \notin V^i_{j'} )} = 1/k~.
 \]
 
 The partitions are chosen independently  using public randomness before receiving the input. For each player $j$, if $i \in S_j$, then they put $V^i_j$ in the stream. Note that the stream consists of $\Theta(m)$ sets. 

If the input is a NO instance, then for each $i \in [m]$, there is at most one set $V^i_j$ in the stream. Therefore, for each element $v \in [n]$ and any collection of $\ell \leq k$ sets $V^{i_1}_{j_1},\ldots, V^{i_\ell}_{j_\ell}$ in the stream,
\begin{align*}
\prob{v \text{ is uniquely covered by }  V^{i_1}_{j_1}, \ldots,V^{i_\ell}_{j_\ell}} & = \ell/k \cdot (1-1/k)^{\ell-1} \leq \ell/k \cdot e^{-(\ell-1)/k}~.
\end{align*}
Therefore, in expectation, $\mu_\ell := \expec{g(\{ V^{i_1}_{j_1}, \ldots,V^{i_\ell}_{j_\ell}\})} \leq \ell/k \cdot e^{-(\ell-1)/k} n$. By an application of Hoeffding's inequality, 
\begin{align*}
\prob{ g ( \{ V^{i_1}_{j_1} \cup \ldots \cup V^{i_\ell}_{j_\ell} \})  > \mu_\ell +  \eps e^{-(k-1)/k} \cdot n  }  & \leq   \exp{ -2 \epsilon^2 e^{-2(\ell-1)/k}n}  \\
& \leq  \exp{-\Omega(\epsilon^2  n)} \leq  \frac{1}{m^{10k}} ~.
\end{align*}

The last inequality follows by letting $n = \Omega(\epsilon^{-2} k \log m)$. The following claim shows that for large $k$, in expectation, picking $k$ sets is optimal in terms of unique coverage.
\begin{lemma}
The function $g(\ell) =  \ell/k \cdot e^{-(\ell-1)/k} n $  is increasing in the interval $ (-\infty,k]$ and decreasing in the interval $[k,+\infty)$. 
\end{lemma}
\begin{proof}
We take the partial derivative of $g$ with respect to $\ell$
\[
\frac{\partial g }{\partial \ell }  = \frac{e^{(1-\ell)/k} (k-\ell)}{ k^2} \cdot n 
\]
and observe that it is non-negative if and only if $\ell \leq k$. 
\end{proof}

By appealing to the union bound over all ${m \choose 1}+\ldots+{m \choose {k-1}}+{m \choose k} \leq O(m^{k+1})$ possible collections $\ell \leq k$ sets, we deduce that with high probability, for all collections of $\ell \leq k$ sets $S_1,\ldots,S_\ell$,
\begin{align*}
g(\{ S_1,\ldots,S_\ell \})  \leq \mu_\ell +  \eps e^{-(k-1)/k} \cdot n & \leq  \ell/k \cdot e^{-(\ell-1)/k} n +  \eps e^{-(k-1)/k} \cdot n  \\
 & \leq (1+\eps) e^{-1+1/k} n~.
\end{align*}

If the input is a YES instance, then clearly, the maximum $k$-unique coverage is $n$. This is because  there exists $i$ such that $i \in S_1 \cap \ldots \cap S_k$ and  therefore $V^i_1,\ldots,V^i_k$ are in the stream and these sets uniquely cover all elements.

Therefore, any constant pass  algorithm that returns better than  a $e^{1-1/k}/(1+\epsilon)$ approximation to $\uc$  for some large enough constant success probability implies a protocol to solve $\Disj(m,k)$. Thus, $\Omega(m/k^2)$ space is required. 
\end{proof}

\subsection{Lower bound for $1+\epsilon$ approximation} 
Assadi \cite{Assadi17} presents a $\Omega(m/\epsilon^2)$ lower bound for the space required to compute a $1+\epsilon$ approximation for $\mc$ when $k = 2$, even when the stream is in a random order and the algorithm is permitted constant passes.  This is proved via a reduction to multiple instances of the Gap-Hamming Distance problem on a hard input distribution, where an input with high maximum coverage corresponds to a YES answer for some Gap-Hamming Distance instance, and a low maximum coverage corresponds to a NO answer for all GHD instances.  This hard distribution has the additional property that high maximum coverage inputs also have high maximum unique coverage, and low maximum coverage inputs have low maximum unique coverage.  Therefore, the following corollary holds:

\begin{corollary}
Any constant-pass randomized algorithm with an approximation factor $1+\epsilon$  for $\uc$ requires $\Omega(m/\epsilon^2)$ space.
\end{corollary}






\section{Handling Insert-Delete Streams}\label{sec:deletion-streams}

\subsection{Proof of Theorem \ref{theorem:deletion-exact}} \label{sec:deletion-streams-exact}
Consider coloring the elements of a universe with a $2$-wise hash-function such that each element is equally likely to get one of $c=10d^2 k$ colors. 

We say a set has color $P$ if the colors of its elements are all different and form the set $P$. Then, via $\ell_0$ sampling \cite{JowhariST11}, use $\tilde{O}(c^d)$ space to sample a set (if one exists) that is colored $P$ (i.e., for each color in $P$ there is exactly one element in the sampled set with this color) for each subset $P\subseteq \{1,2,\ldots,c\}$ of size at most $d$. 

\begin{definition}
Let $C$ be a collection of at most $k$ sets where each set have size at most $d$. Say a set $S$ in $C$ is \emph{good} with respect to $C$ if the elements of $S$ receive different colors and they are all different from the colors received by elements in $(\cup_{S' \in C} S')\setminus S$. 
\end{definition}

For any good set $S$ in the collection, let  $r(S)$ be the set found by the sampling algorithm that is colored the same as set $S$. We call $r(S)$ the \emph{replacement} for $S$.

\begin{lemma}\label{lem:rep}
Removing  sets $S_1, S_2,\ldots, S_g$ that are good with respect to (w.r.t.) $C$ from $C$ and replacing them by $r(S_1), r(S_2),\ldots, r(S_g)$ yields a new collection that (uniquely) covers at least the same number of elements as $C$.
\end{lemma}
\begin{proof}
Let $R_0$ be the set of colors used to color elements in $\cup_{i=1}^g S_i$ and let $R_1$ be the set of colors used to color elements in $(\cup_{S' \in C} S')\setminus \left( \cup_{i=1}^g S_i\right )$. Because  $S_1, S_2,\ldots, S_g$ are good sets, $|R_0|=|\cup_{i=1}^g S_i|$ and $R_0\cap R_1=\emptyset$. After replacing $S_1, S_2,\ldots, S_g$ by $r(S_1), r(S_2),\ldots$, the multiplicity of an element with a color in $ R_1$ is unchanged. For any color in $R_0$, let $e$ be the element in $\cup_{i=1}^g S_i$ with this color. There will be at least one element with the same color as $e$ after the collection is transformed. It follows that the coverage of the collection does not decrease: the removal of $S_1, S_2,\ldots, S_g$ reduces the coverage by at most $|\cup_{i=1}^g S_i|$ but adding $r(S_1), r(S_2),\ldots$  increases the coverage by at least $|R_0|$. To argue that the unique coverage of the collection does not decrease, note that if $e$ had multiplicity 1 then the element with the same color as $e$ after the transformation also has multiplicity 1.

\end{proof}

\begin{lemma}\label{lem:prob}
For any $C'\subseteq C$, $\Pr[\mbox{number of good sets in $C'$ is $\geq 4|C'|/5$}] \geq 1/2$.
\end{lemma} 
\begin{proof}
First note that, a set is not good if one of its element shares a color with an element in that set or in another set in the collection. By the union bound,
\[\Pr[\mbox{set is not good}] \leq d(dk)/c=1/10 \ . \]
Hence, for any subset $C'$ of $C$, $\expec{\mbox{number of bad sets in $C'$}} \leq  |C'|/10$
and the lemma follows via Markov inequality.
\end{proof}

\begin{theorem}After repeating the random coloring and sampling $O(\log k)$ times, we have a collection of sets that includes the collection of size at most $k$ that (uniquely) covers the maximum number of elements.
\end{theorem}

\begin{proof}
For the sake of analysis, let $C_0$ be a collection of at most $k$ sets with optimum (unique) coverage. Let $C'=C_0$. 
\begin{enumerate}
\item Randomly color elements. Let $C_1$ be the collection formed from $C_0$ by replacing all sets in $C_0$ that are good sets wrt $C_0$ by their replacements. Remove all good sets (w.r.t. $C_0)$ from $C'$.
\item Randomly color elements. Let $C_2$ be the collection formed from $C_1$ by replacing all sets in $C'$ that are good sets wrt $C_1$ by their replacements.  Remove all good sets (w.r.t. $C_1)$ from $C'$.
\item \ldots continue in this way for $O(\log k)$ steps.
\end{enumerate}
In each step, the size of $|C'|$ decreases by a constant factor  with constant probability by appealing to Lemma \ref{lem:prob}. Hence after $O(\log k)$ steps $|C'|=0$. Note that the (unique) coverage of $C_{O(\log k)}$ is at least the (unique) coverage of $C_0$ by Lemma \ref{lem:rep}. 
\end{proof}
Noting that the $O(\log k)$ colorings/sampling can be performed in parallel, we have a single-pass algorithm.

\subsection{Handling deletions for the algorithm in  Theorem \ref{thm:alg6}} \label{sec:deletion-streams-2}
We now explain how the approach using in Theorem \ref{thm:alg6} can be extended to the case where sets may be inserted and deleted. In this setting, it is not immediately obvious how to select the largest $\roundup{r k/\epsilon}$ sets; the approach used when sets are only inserted does not extend. Note that in this model we can set $m$ to be the maximum number of sets that have been inserted and not deleted at any prefix of the stream rather than the total number of sets inserted/deleted.

However, we can extend the result as follows. Suppose the sketch of a set for approximating maximum (unique) coverage requires $B$ bits; recall from Section \ref{sec:sketch} that  $B=k\epsilon^{-2}\polylog(n,m)$ suffices. We can  encode such a sketch of a set $S$ as an integer $i(S)\in [2^B]$. Suppose we know that exactly $\roundup{r k/\epsilon}$ sets have size at least some threshold $t$. We will remove this assumption shortly. Consider the vector $x\in [N]$ where $N=2^B$ that is initially 0 and then is updated by a stream of set insertions/deletions as follows:
\begin{enumerate}
\item When $S$ is inserted, if $|S|\geq t$, then $x_{i(S)}\leftarrow x_{i(S)}+1$.
\item When $S$ is deleted, if $|S|\geq t$, then $x_{i(S)}\leftarrow x_{i(S)}-1$.
\end{enumerate}
At the end of this process $x\in \{0,1, \ldots, ,m\}^{2^B}$, $\ell_1(x)=\roundup{r k/\epsilon}$, and reconstruct the sketches of largest $\eta k$ sets given $x$. Unfortunately, storing $x$ explicitly in small space is not possible since, while we are promised that at the end of the stream $\ell_1(x)=\roundup{r k/\epsilon}$, during the stream it could be that $x$ is an arbitrary binary string with $m$ one's and this requires $\Omega(m)$ memory to store. To get around this, it is sufficient to maintain a linear sketch of $x$ itself that support sparse recovery. For our purposes,  the CountMin Sketch \cite{CormodeM05} is sufficient although other approaches are possible. The CountMin Sketch allows $x$ to be reconstructed with probability $1-\delta$ using a sketch of size 
\[O(\log N+\roundup{r k/\epsilon} \log(\roundup{r k/\epsilon} /\delta)\log m)=O(\roundup{r k/\epsilon} \epsilon^{-2}\polylog(n,m) ) \ .\]

To remove the assumption that we do not know $t$ in advance, we consider values:
\[t_0, t_1, \ldots , t_{\lceil \log_{ 1+\epsilon} m\rceil } \mbox{ where } t_i=(1+\epsilon)^i \ .\] We define vector $x^0, x^1, \ldots \in \{0,1, \ldots, ,m\}^{2^B}$ where $x^i$ is only updated when a set of size $\leq t_i$ but $>t_{i-1}$ is inserted/deleted. Then there exists $i$ such that $\leq \roundup{r k/\epsilon}$ sets have size $\leq t_{i-1}$ and the sketches of these sets can be reconstructed from $x^0, \ldots, x^{t_{i-1}}$. To ensure we have $\roundup{r k/\epsilon}$ sets, we may need some additional sketches corresponding to sets of size $>t_{i-1}$ and $\leq t_i$ but unfortunately there could be $m$ such sets and we are only guaranteed recovery of $x^{t_i}$ when it is sparse. However, if this is indeed the case we can still recover enough entries of $x^{t_1}$ by first subsampling the entries at the appropriate rate (we can guess sampling rate  $1, 1/2, 1/2^2, \ldots 1/m$) in the standard way. Note that we can keep track of $\ell_1(x^i)$ exactly for each $i$ using $O(\log m)$ space.



\section{The Subsampling Framework} \label{appendix:subsampling}

Assuming we have $v$ such that $\opt/2 \leq v \leq \opt$. Let $h:[n] \rightarrow \{0,1\}$ be a hash function that is $\Omega(\epsilon^{-2} k \log m)$-wise independent. We run our algorithm on the subsampled universe $U' = \{ u \in U: h(u)=1\}$. Furthermore, let

\[
\prob{h(u)=1} = p = \frac{c k \log m}{\epsilon^{2} v}
\] 
where $c$ is some sufficiently large constant.  Let $S' = S \cap U'$ and let $\opt'$ be the optimal unique coverage value in the subsampled set system. The following result is from McGregor and Vu \cite{MV18}. We note that the proof is the same except that the indicator variables now correspond to the events that an element being uniquely covered (instead of being covered).
\begin{lemma} \label{lem:sampling}
With probability at least $1-1/\poly(m)$, we have that
\[
p \opt (1+\eps) \geq \opt' \geq p \opt (1-\epsilon)
\]
Furthermore, if $S_1,\ldots,S_k$ satisfies $g(\{ S'_1,\ldots,S'_k \}) \geq p\opt (1-\epsilon)/t $ then 
\[
g(\{S_1,\ldots,S_k\}) \geq \opt (1/t -2 \epsilon) ~.
\]
\end{lemma}

We could guess $v =1,2,4,\ldots,n$. One of the guesses must be between $\opt/2$ and $\opt$ which means $\opt' = O(\epsilon^{-2}k \log m)$. Furthermore, if we find a $1/t$ approximation on the subsampled universe, then that corresponds to a $1/t-2\epsilon$ approximation in the original universe. We note that as long as $v \leq \opt$ and $h$ is $\Omega(\epsilon^{-2}k \log m)$-wise independent,  we have (see \cite{SchmidtSS95}, Theorem 5):
\begin{align*}
& \prob{g(\{S_1',\ldots,S_\ell'\} ) = p  \cdot g(\{S_1,\ldots,S_\ell \} ) \pm \epsilon p \opt } \\
 & \geq  1 - \exp{-\Omega( k \log m)} \geq  1-1/m^{\Omega(k)}~.
\end{align*}
 This gives us Lemma \ref{lem:sampling} even for when $v < \opt/2$. However, if $v \leq \opt/2$, then $\opt'$ may be larger than $O(\epsilon^{-2} k \log m)$, and we may use too much memory. To this end, we simply terminate those instantiations. Among the instantiations that are not terminated, we return the solution given by the smallest guess.

\bibliography{references}

\end{document}